\documentclass[aps, prl,  a4paper, showpacs, twocolumn, english, 10pt]{revtex4-2}
\usepackage{bbm, amsthm, bm, textcomp, nicefrac, geometry,ragged2e}
\usepackage[dvipsnames]{xcolor}
\usepackage{float}
\usepackage[bbgreekl]{mathbbol}
\usepackage{graphicx,epstopdf,color,verbatim,enumitem}
\usepackage{pbox,array}
\usepackage{mathrsfs}
\usepackage{amssymb}
\usepackage{textgreek}
\usepackage{mathtools}
\usepackage{stackrel}
\usepackage[thinlines]{easytable}
\usepackage{amsmath}
\newtheorem{theorem}{Theorem}

\usepackage[utf8]{inputenc}
\makeatletter
\usepackage{soul}
\usepackage[caption=false]{subfig}
\usepackage{hyperref}
\hypersetup{
    colorlinks = true,
    linkcolor = teal,
    citecolor = teal,
    filecolor = teal,      
    urlcolor = teal,
}

\usepackage[T1]{fontenc}
\usepackage[caption=false]{subfig}
\usepackage{babel}
\usepackage[linesnumbered,ruled]{algorithm2e}
\usepackage[normalem]{ulem}

\addto\captionsenglish{}

\newcommand{\ket}[1]{\left|#1\right\rangle}
\definecolor{brickred}{rgb}{0.8, 0.0, 0.0}

\begin{document}

%\title{A fission protocol for graph states}
\title{Graph State Fission}

% \author{Ferran Riera-S\`abat}
% \thanks{These two authors contributed equally } 

\author{Jorge Miguel-Ramiro}
%\thanks{These two authors contributed equally }

\author{Wolfgang D\"ur}
\affiliation{Universit\"at Innsbruck, Institut f\"ur Theoretische Physik, Technikerstra{\ss}e 21a, Innsbruck 6020, Austria}

\date{\today}
\begin{abstract}
Graph states are a fundamental entanglement resource for multipartite quantum applications which are in general challenging to transform efficiently. While fusion operations for merging entangled states are well-developed, no direct protocol exists for the reverse process, which we term “fission”. We introduce a simple, yet powerful, protocol that achieves this, allowing a qubit to split while preserving selective connections with minimum entanglement overhead. This tool offers flexible entanglement management with potential applications in secure communication, error correction and adaptive entanglement distribution.

\end{abstract}

\maketitle

%\tableofcontents

\textbf{\textit{Introduction.--}} Entanglement is a fundamental resource for quantum computation and quantum communication, underpinning key applications such as teleportation \cite{Pirandola_2015, Ren_2017, Hu_2023}, quantum sensing \cite{Giovannetti_2011, Kessler2014, Sekatski2020}, distributed quantum computing \cite{CiracDistributed, Hayashi15, Cacciapuoti2020}, or quantum cryptography \cite{Gisin2002, Pirandola_2020}. As quantum technologies advance, the generation, distribution, and manipulation of multipartite entangled states becomes ever more crucial \cite{Horodecki2009, Walter_2016, Pirker_2019, Meignant2019, Navascues2020, Zhong_2021, Miguel_Ramiro_2023, Fan2024}. Among these, graph states stand out as a prominent form of multipartite entanglement \cite{Gottesman1997, Hein2004, Hein2006}, fundamental to the development of scalable quantum networks. Manipulating and transforming graph states, however, remains a significant challenge. Achieving desired transformations often requires substantial entanglement resources to modify connections between qubits, making efficient state manipulation complex and resource-demanding.

One well-studied operation in this context is graph state fusion \cite{Knill_2001, Browne2005,zdemir_2011}, which merges two qubits to form larger or more entangled structures. Quantum fusion has obtained considerable interest due to its practical applications, particularly in quantum repeater protocols \cite{Dur99, Azuma2023, Miguel2023, Thomas_2024}, quantum computing \cite{Browne2005, yoran2003, Bartolucci_2023, Paesani2023}, or the generation and growth of multipartite states in quantum networks \cite{Fischer_2021, Bugalho_2023, Avis2023, rimock2024}. Various experimental proposals and implementations \cite{Pan98, Zhang2006} have further supported its role in enhancing entanglement distribution across distributed systems.

Here, we introduce the reverse operation, which we denote as ``fission'', where a single qubit is split into multiple qubits to create smaller and selectively connected substructures. This task has not been directly addressed, instead, straightforward approaches like edge and vertex removal via qubit measurements are typically used. These methods, however, are often suboptimal, introducing larger entanglement costs or disrupting the states underlying structure. We introduce a graph state fission protocol that allows for partitioning of a graph state qubit, preserving selected connections with minimal additional entanglement. The protocol utilizes an auxiliary Bell or GHZ state shared among specific neighbors to retain desired connections, achieving selective fission while maintaining the integrity of the graph entanglement structure. This approach enhances both resource efficiency and security by limiting interactions to only the qubits involved in the fission process. Furthermore, we demonstrate that this protocol is optimal in the basic scenario in terms of entanglement resources required, offering new capabilities for entanglement manipulation in quantum networks and secure quantum communication.

% The paper is structured as follows. We review some basic concepts in Sec.~\ref{sec:back}, while we introduce and explain the fission protocol in Sec.~\ref{sec:fission}. We discuss potential applications for graph state fission in Sec.~\ref{sec:apps} and we conclude in Sec.~\ref{sec:conclusions}.

% \textbf{\textit{Background.--}} \label{sec:back} 

\textbf{\textit{Basic concepts.--}} We begin by briefly reviewing some basic concepts and notation used throughout this paper.  

Graph states \cite{Gottesman1997, Hein2004, Hein2006}, are multi-qubit quantum states represented by a graph $G=(V,E)$, where vertices V correspond to qubits, and edges E represent entanglement between connected qubits. Defined as the unique +1 eigenstates of stabilizer operators \cite{Gottesman1997} associated with each qubit, $K_{a} = \sigma_{x}^{(a)} \prod_{(a,b) \in E} \sigma_{z}^{(b)}$,
% \begin{align}
% K_{a} = \sigma_{x}^{(a)} \prod_{(a,b) \in E} \sigma_{z}^{(b)}, \label{def:graph}
% \end{align}
for all $a \in V$, graph states serve as a fundamental resource in quantum computing \cite{Rauss2001, Nielsen04}, especially in one-way computation and entanglement-based protocols \cite{Browne2005, Walther_2005, Epping_2016, Pirker_2018, Miguel_Ramiro_2021}.

Important instances of graph states include two-qubit Bell states and $m$-qubit GHZ states, given by $\ket{\Phi^+} = \frac{1}{\sqrt{2}}\left( \ket{00}+ \ket{11} \right)$ and $\ket{\mathrm{GHZ}_m} = \frac{1}{\sqrt{2}}\left( \ket{0}^{\otimes m} + \ket{1}^{\otimes m} \right)$
% \begin{align}
% &\ket{\Phi^+} = \frac{1}{\sqrt{2}}\left( \ket{00}+ \ket{11} \right), \\ 
% &\ket{\mathrm{GHZ}_m} = \frac{1}{\sqrt{2}}\left( \ket{0}^{\otimes m} + \ket{1}^{\otimes m} \right)
% \end{align}
respectively, up to local unitaries.

A fundamental transformation in graph states is local complementation \cite{Hein2004, Hein2006}, a unitary operation that modifies the connections among certain qubit neighbors, leading to an equivalent graph state. Furthermore, Pauli measurements on a qubit of a graph state invoke local complementations on some of the qubits, followed by a deletion of all the incident edges to the qubit \cite{Hein2004, Hein2006}.  

Moreover, we use the concept of ebit of entanglement \cite{nielsen2002quantum, Horodecki2009}, defined as the unit of bipartite entanglement, i.e., the amount of entanglement contained in a Bell state. This allows us to evaluate the amount of entanglement between bipartitions $(A,B)$ of a graph state, where $A$ and $B$ are disjoint subsets of vertices of the graph, such that $A \cup B= V$.

% \textit{Notation.--} Throughout this work, we use the terms "split" and "separate" synonymously with "fission." Similarly, we indistinctly refer to each constituent of the graph as a "party" or "node," indicating entities where local operations can be applied—-each of which may hold more than one qubit.

% \subsection{Related work}
% related work- fussion

\textbf{\textit{Graph state fission.--}} We propose a protocol enabling the separation of any graph state through a fission process on a selected qubit, preserving all existing edges, as shown in Fig.~\ref{fig:fission0}. The protocol utilizes a minimal amount of additional entanglement (that can be proven optimal in its basic configuration, detailed below) and ancillary qubits, ensuring selective connectivity between the neighbors of the split qubits. Crucially, this protocol incorporates inherent security, requiring only the collaboration of qubits directly involved in the fission process.
\begin{figure}
    \centering
    \includegraphics[width=0.95\columnwidth]{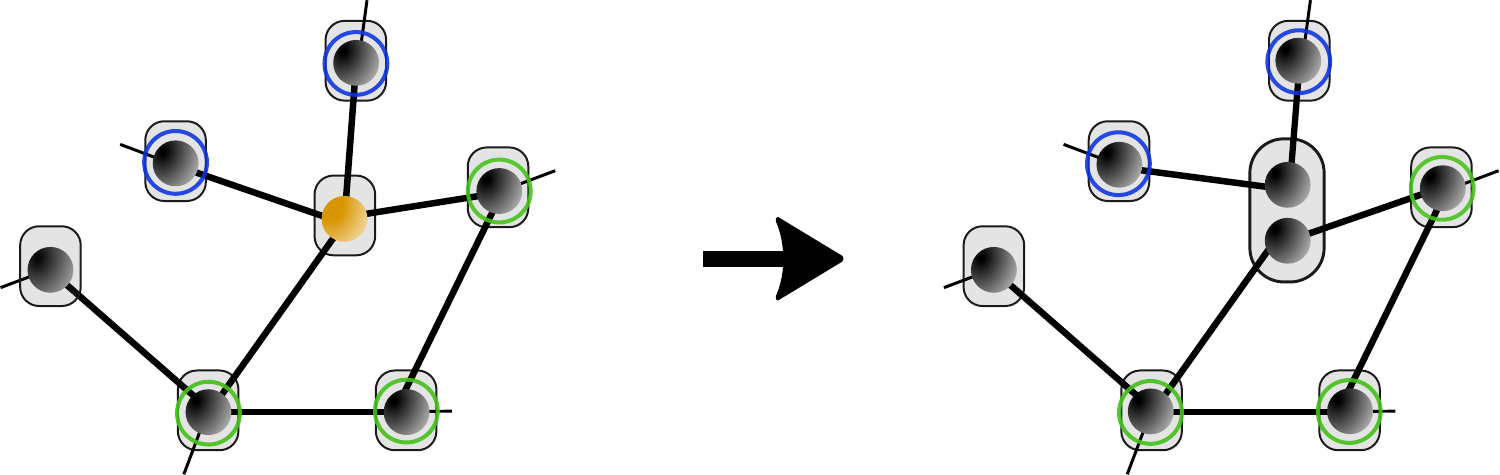}
    \caption{Graph state fission. Given a selected qubit (yellow) on an arbitrary graph, the protocol achieves to split such qubit into two, keeping all edges while choosing the neighbors (blue or green) connected to each of the remaining split qubits. Additional entanglement is required to accomplish this task.}
    \label{fig:fission0}
\end{figure}

The manipulation and transformation of graph states are well-studied challenges in quantum information theory \cite{Hein2004, Hein2006}, with general approaches shown to be NP-complete \cite{Dahlberg2018,Dahlberg2020}. In our scenario, achieving fission without supplementary entanglement would be impossible, as the resulting graph states typically exhibit increased entanglement features.

Existing tools \cite{Hein2004, Hein2006, Meignant2019,Hahn_2019, Fischer_2021, Miguel_Ramiro_2023} for solving the fission problem would typically rely on edge deletion via qubit measurements, followed by entanglement regeneration. A Pauli Z-measurements of a qubit $a$ leads to a deletion of all edges that include $a$ \cite{ Hein2006,Hahn_2019}. This leads to undesirable additional entanglement losses and increased regeneration demands, and results in higher overheads compared to our protocol. Moreover, these methods would frequently compromise security, as they require interactions involving qubits beyond those directly participating in the fission process, a drawback that our protocol avoids by design.

While our fission protocol may appear simple, it is, in fact, a non-trivial tool that enables precise, selective fission operations on any qubit within a graph state.

\subsubsection{Fission protocol 1. One neighbor} 
We introduce our proposed fission protocol by first addressing the simplest case, where only one of the original neighbors is carried along with the split subgraph,  and then generalizing to more complex scenarios. We demonstrate the optimality of this approach regarding the minimal additional entanglement required to complete it.
\begin{figure*}
    \centering
    \includegraphics[width=0.85\textwidth]{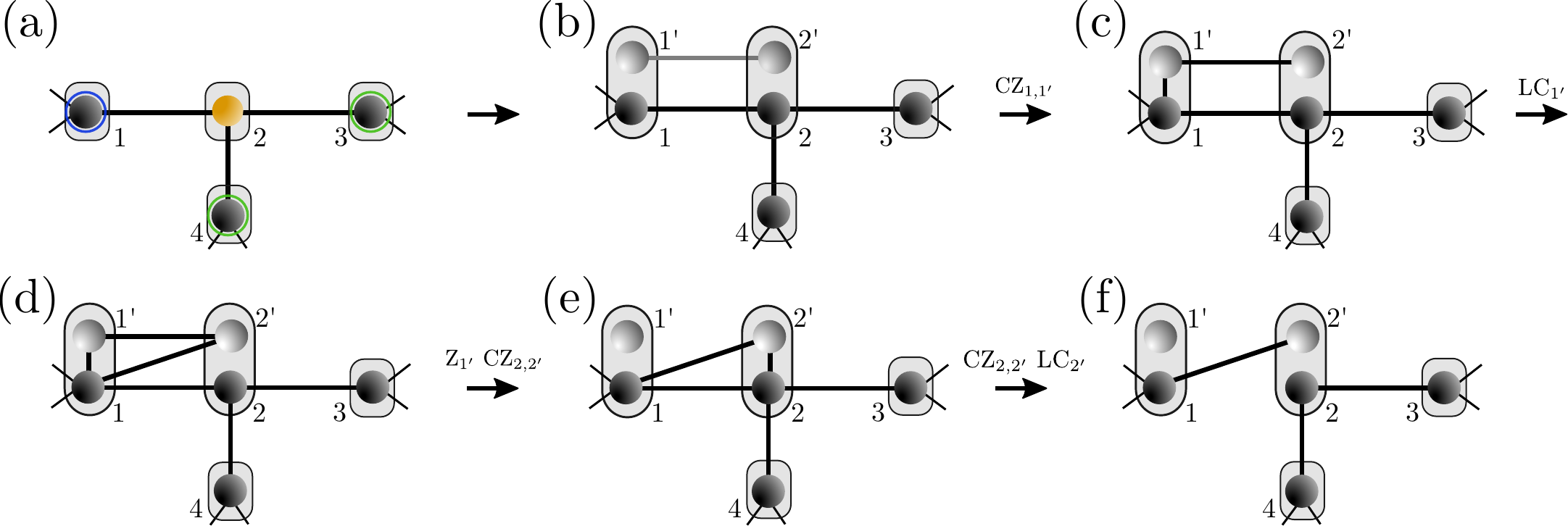}
    \caption{Illustration of the fission protocol 1: qubit $2$ is split, with qubit $1$ maintained as a neighbor, resulting in two subgraph states, disconnected from each other in case the neighborhoods were not originally connected. An auxiliary Bell state provides the minimal entanglement needed for this process. }
    \label{fig:fission1}
\end{figure*}

Fig.~\ref{fig:fission1} illustrates the steps in our fission protocol for the simplest splitting case. Starting with an arbitrary graph state in (a), we aim to split the qubit at node 2 into two qubits, such that its neighbor at node 1 remains connected only to the newly formed qubit.

An auxiliary Bell state (1',2') is prepared between nodes 1 and 2, ensuring qubits (1,1') and (2,2') are co-located for joint operations, see Fig.~\ref{fig:fission1} (b). Qubits 1 and 1' are then locally entangled using controlled-Z operations (c), followed by a local complementation operation on qubit 2' (d). Then, qubit 1' is measured in the Pauli Z basis to remove its edges, while an entangling operation is applied between qubits 2 and 2' (e). Finally, a local complementation on qubit 2', followed by disentangling operation on qubits 2 and 2', completes the protocol (f).

The result is two independent subgraph states (in case the neighborhoods were not originally connected to each other). Qubit 1 is now connected to qubit 2' in one of the subgraphs, while every other connection remains unaffected. %This process achieves a clean fission with optimal entanglement usage and without disturbing the rest of the graph.
Only active participation of nodes 1 and 2 is required, while their second and subsequent neighborhoods play no role in the process.

Likewise, the qubit can be further partitioned into a larger number of elements using the same method, by simply introducing an auxiliary Bell state for each neighboring pair intended for separation, without removing any connections or modifying the intrinsic structure of the original graph, apart from the fission process.

\textit{Optimality.--} The fission approach introduced above, Fig.~\ref{fig:fission1}, is optimal regarding the additional entanglement required for their execution.

\begin{theorem}
The minimum amount of entanglement necessary to apply one fission to any qubit in a connected graph state, carrying along one of its original neighbors is, at least, one ebit.
\end{theorem}

\begin{proof}
Consider the scenario treated in Fig.~\ref{fig:fission1}. Note first that a connected graph state is a 1-uniform state \cite{Hein2006,Sudevan2022}, meaning that the reduced density operator of a single qubit is always maximally mixed, indicating it shares precisely one ebit of entanglement with the rest. After the fission process, however, two qubits occupy the central node. By the same 1-uniformity argument, this new two-qubit party shares two ebits of entanglement with the remaining nodes.

According to the fact that local operations alone cannot increase entanglement across any bipartition \cite{Hein2004}, we conclude that, at least, introducing one additional ebit of entanglement is required for the fission operation to maintain the entanglement features between the central node and the rest of the graph.
\end{proof}

The optimality result follows directly from the fact that our protocol introduces just a single Bell pair, which provides precisely one ebit of entanglement, the minimal amount required for the fission process.

A similar argument demonstrates that in order to perform $k$ fission operations to a certain qubit, at least $k$ additional ebits of entanglement are required.

\subsubsection{Fission protocol 2. Arbitrary neighbors}
We now turn our attention to a more general scenario, where, given an arbitrary graph state, a fission operation is applied to one of the qubits, enabling the selection of specific neighbors to remain connected with the split qubit, wihtout removing the rest of the edges. The protocol, closely related to the one presented above, is depicted in Fig.~\ref{fig:fission2}, where we now utilize a $m$-qubit GHZ state. The size of the GHZ state corresponds to the minimum size of each set of neighbors that are split along the fission qubit. 
\begin{figure*}
    \centering
    \includegraphics[width=0.8\textwidth]{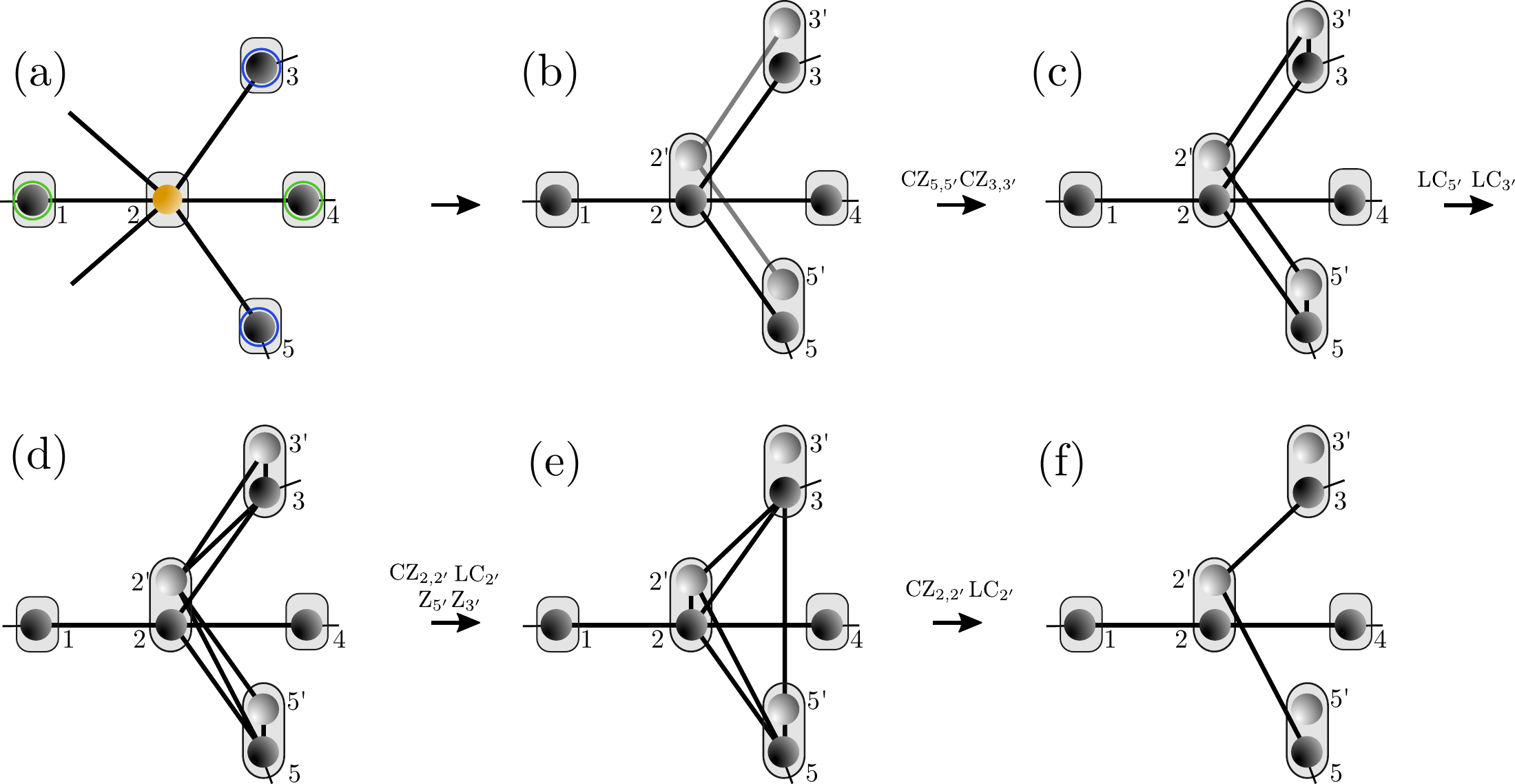}
    \caption{Fission protocol 2 for arbitrary neighbor selection: qubit $2$ is split, carrying along qubits $3$ and $5$ as its connected neighbors, and leaving the rest of the graph untouched. An additional GHZ state is required for this task, shared between the desired partitioned neighbors ($2,3,5$). }
    \label{fig:fission2}
\end{figure*}

Consider an arbitrary graph, Fig.~\ref{fig:fission2} (a) where a selected qubit (2) shares graph edges with its neighbors. We omit other possible edges for illustrative purposes. The objective is to split qubit 2 into two, such that certain chosen initial neighbors (3 and 5) are carried along with the new qubit, while the other edges remain with the original qubit. For that purpose, we use a GHZ state (b), shared between the involved parties (a three-qubit GHZ between nodes 2, 3 and 5 in our case). Entangling operations are then performed within nodes 2 and 5, linking the main graph state to the auxiliary GHZ state. Subsequent local complementations are applied in particles 3' and 5' (d), followed by local Z measurements of such qubits (e). Entangling qubits 2 and 2' followed by a local complementation of particle 2' and a disentangling operation (f), completes the protocol. This process yields a graph state, with qubits 3 and 5 connected to the newly formed qubit, while the rest of the graph remains unaffected.

This fission strategy is fully general, allowing for the selective application of fission operations on any qubit within a graph state while choosing which neighbors remain connected to each of the resulting split qubits. Additionally, the fission process can be repeated multiple times, enabling the division of a single qubit into $k$ particles. This flexibility significantly enhances the capacity for manipulation of multipartite entangled states.

Importantly, we remark that only the qubits directly involved in the fission process actively participate, maintaining security and limiting interactions to essential or trusted parties.

\textit{Minimizing resources.--}  Despite the generality of this process, the use of a GHZ state does not always guarantee minimal resource consumption. Depending on the specific connectivity of the graph state, it is often possible to identify locally equivalent graph states that reduce the entanglement overhead. 

A simple example is shown in Fig.~\ref{fig:fission3}. Here, the fission protocol aims to separate a central node (c) while retaining three neighbors with the split qubit. Given the structure of the initial graph state, it is possible to find a local unitary (LU) equivalent state in which the neighborhood of the central node is reduced, yet preserves the structures designated for fission. This adjustment minimizes the required entanglement from a four-qubit GHZ state to just one Bell pair (as in Protocol 1), achieving the desired separation with lower resource overhead. Notice that the resulting graph states are LU equivalent, and in fact local complementation on all vertices except the central node (c) can performed to minimize the local degree of (c) and hence the required size of GHZ state to perform fission. 
\begin{figure}
    \centering
    \includegraphics[width=0.9\columnwidth]{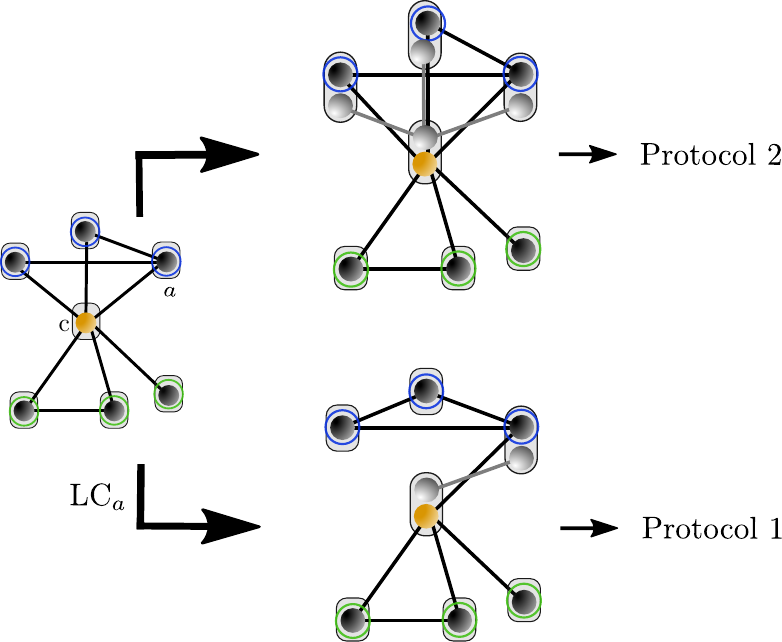}
    \caption{Example. Given this graph state where the blue neighborhood should be part of the split qubit, one can directly apply protocol 2 using a four-qubit GHZ state, Fig.~\ref{fig:fission2}, or find a local equivalent graph that allows to apply fission protocol 1, Fig.~\ref{fig:fission1}, therefore optimizing resources.}
    \label{fig:fission3}
\end{figure}
By leveraging local operations, we can thus optimize the protocol resource demands, enhancing its overall efficiency.

\textbf{\textit{Applications.--}} We now discuss possible applications and benefits of the fission strategies introduced above. In some situations, the synergic combination of fission and fusion techniques can lead to significant benefits. Some examples include: 

\textit{Improving other protocols performance.--}
Numerous protocols in quantum computation and communication are sensitive to the size of the quantum states involved. For example, entanglement purification protocols \cite{Bennett1996, Deutsch_1996, Dur07} are typically more efficient when applied to smaller states, compared to larger or arbitrary graph states \cite{Dur03, Dur06, Dur07}. The performance gap between small and complex states is considerable, impacting both applicability and efficiency \cite{Dur03, Dur06, Dur07}. The ability to partition larger states into smaller, manageable subgraphs (without compromising their structure), purify them, and subsequently reassemble them via fusion operations could be significantly beneficial. 

In a similar spirit, in modular quantum architectures \cite{Monroe2014, Pirker_2018, Bombin21}, isolating subgraphs without losing the entanglement structure can improve fault tolerance processes, both in quantum computation \cite{PRESKILL_1998, Katabarwa2024} and communication \cite{Childress2006, Mural2014}. By using fission to isolate noisy qubits or manage errors locally, error correction can be applied while keeping other qubits unaffected. This approach could allow modular quantum computers to handle errors more effectively, also potentially avoiding the unwanted effects of noise crosstalk \cite{Murali_2020}.  By maintaining entanglement only within necessary subgraphs, the protocol aids in local error correction and limits the spread of noise across a distributed system, which is critical for fault-tolerance quantum computing.

\textit{Entanglement manipulation, distribution and routing.--}
Distribution of entanglement among distant parties is a key subject of study in quantum networks \cite{Avis2023, Bugalho_2023, Khan_2019, Cacciapuoti_2024, Zurita2024}. The ability to split graph states while retaining key entangled connections can enable a more flexible distribution of such entanglement. By allowing selective graph state separation, the fission protocol can help to manage entanglement resources dynamically, optimizing the network performance without requiring additional complex state preparation for each task. This capability can be especially valuable in hierarchical or multi-layered network structures \cite{Pirker_2019, Miguel_Ramiro_2023, Xu_2009, Davarzani_2022}, where managing entanglement across various levels or clusters is required.

A quantum network might also have fixed structures for the generation and distribution of multipartite entanglement.  Assisted with fission operations, one can generate and establish different kinds of states from the fixed structures provided by the network, without the need to generate entanglement in different ways at the network level.  In the context of entanglement-based quantum networks \cite{Pirker_2018, Pirker_2019}, the fission tool provides an a flexible possibility of adapting entanglement resources upon network demand.

In terms of entanglement routing \cite{Hahn_2019, Pant_2019, Li_2021, Zeng_2022}, fission can provide a method for creating on-demand entanglement paths, optimizing these paths without constantly regenerating entangled states. For instance, if only specific network paths are required, fission can isolate these paths from the larger entangled structure, thereby optimizing resource allocation across the network, while preserving the overall remaining network state.

\textit{Security aspects.--}
Graph state fission could enhance quantum secret sharing \cite{Hillery99, Xiao2004} and other multiparty cryptographic protocols \cite{Epping_2017, Zhou18, Memmen2023} by allowing certain entangled nodes to retain connections only with selected parties. This selective entanglement control could be useful for secure protocols, where specific subgroups must remain entangled while ensuring the rest of the network is isolated. This application could be extended to more general scenarios, where parts of a network are not trustworthy and can be disconnected from the rest, without their collaboration nor losing structure, by fission means. Given the inherent security of our protocols, where only involved parties in the fission process need to collaborate to separate others, one can expect to enhance both flexibility and security in multiparty cryptographic protocols.

In summary, the fission tool introduced in this work has the potential to significantly enhance quantum network management, improve efficiency in entanglement-based applications, and provide robust, flexible solutions for network and modular quantum systems. Its applications can be applied in various areas where efficient and dynamic manipulation of entangled states is essential for scalability, performance, and resource optimization.

%\section{Fusion device?}

\textbf{\textit{Conclusions.--}} In this work, we have introduced an optimal protocol for applying fission to a qubit within an arbitrary graph state, allowing it to be split into two —or more— qubits, with selective control over the neighbors retained by each resulting qubit. We have shown that our approach is resource-efficient, requiring minimal additional entanglement, and that it surpasses straightforward methods in terms of efficiency, security, and overhead.

This fission protocol represents a novel tool in the manipulation, transformation, and distribution of multipartite entangled states, potentially enhancing applications across quantum communication, cryptography, and modular computing. 

Future research could explore extensions of this approach to other classes of entangled states beyond graph states, examining how these generalized fission operations could broaden the scope of potential applications. Additionally, investigating transformations under local unitaries may increase the protocol flexibility and adaptability in various quantum architectures. Further analysis of practical applications and experimental implementations will be important for realizing the full potential of this technique.

\textbf{\textit{Acknowledgments.--}} This research was funded in whole or in part by the Austrian Science Fund (FWF) 10.55776/P36009 and 10.55776/P36010. For open access purposes, the author has applied a CC BY public copyright license to any author-accepted manuscript version arising from this submission.

\bibliographystyle{apsrev4-2}
\bibliography{Fission}

\clearpage

% \onecolumn\newpage
% \appendix

\end{document}